\documentclass[journal]{IEEEtran}

\usepackage{cite}

\usepackage[pdftex]{graphicx}
\graphicspath{{../Drawings/}{./}{./Drawings/}}
\DeclareGraphicsExtensions{.pdf,.jpeg,.jpg,.png,.eps}
\ifCLASSOPTIONcompsoc
\usepackage[caption=false,font=normalsize,labelfon t=sf,textfont=sf]{subfig} \else \usepackage[caption=false,font=footnotesize]{subfig} 
\fi

\usepackage[usenames, dvipsnames]{xcolor}

\usepackage{amsmath,amssymb}

\newtheorem{theorem}{Theorem}
\newtheorem{corollary}{Corollary}
\newtheorem{proposition}{Proposition}

\newtheorem{definition}{Definition}
\newtheorem{proof}{Proof}
\newtheorem{example}{Example}

\usepackage{algorithmic}

\usepackage{array}

\ifCLASSOPTIONcompsoc
 \usepackage[caption=false,font=normalsize,labelfont=sf,textfont=sf]{subfig}
\else
 \usepackage[caption=false,font=footnotesize]{subfig}
\fi

\usepackage{url}
\begin{document}
\title{Lattice Functions for the Analysis of Analog-to-Digital Conversion}
\author{Pablo~Mart\'inez-Nuevo,~\IEEEmembership{Member,~IEEE,}
        and~Alan~V.~Oppenheim,~\IEEEmembership{Life Fellow,~IEEE}
\thanks{\copyright~2019 IEEE. Personal use of this material is permitted. Permission from IEEE must be obtained for all other uses, in any current or future media, including reprinting/republishing this material for advertising or promotional purposes, creating new collective works, for resale or redistribution to servers or lists, or reuse of any copyrighted component of this work in other works.

This work was supported in part by the Texas Instruments Leadership University Program and in part by support from Analog Devices, Inc.. The work of P. Mart\'inez-Nuevo was also supported in part by Fundaci\'on Rafael del Pino, Madrid, Spain.

P. Mart\'inez-Nuevo was with the Department of Electrical Engineering and Computer Science, Massachusetts Institute of Technology. He is now with the research department at Bang \& Olufsen, 7600 Struer, Denmark (e-mail: pmnuevo@alum.mit.edu).

A. V. Oppenheim is with the Department of Electrical Engineering and Computer Science, Massachusetts Institute of Technology, Cambridge MA 02139 USA (e-mail: avo@mit.edu).}
\thanks{Digital Object Identifier 10.1109/TIT.2019.2907996}}

\markboth{IEEE Transactions on Information Theory}%
{Shell \MakeLowercase{\textit{et al.}}: Bare Demo of IEEEtran.cls for IEEE Journals}

\maketitle

\begin{abstract}
Analog-to-digital (A/D) converters are the common interface between analog signals and the domain of digital discrete-time signal processing. In essence, this domain simultaneously incorporates quantization both in amplitude and time, i.e. amplitude quantization and uniform time sampling. Thus, we view A/D conversion as a sampling process in both the time and amplitude domains based on the observation that the underlying continuous-time signals representing digital sequences can be sampled in a lattice---i.e. at points restricted to lie on a uniform grid both in time and amplitude. We refer to them as lattice functions. This is in contrast with the traditional approach based on the classical sampling theorem and quantization error analysis. The latter has been mainly addressed with the help of probabilistic models, or deterministic ones either confined to very particular scenarios or considering worst-case assumptions. In this paper, we provide a deterministic theoretical analysis and framework for the functions involved in digital discrete-time processing. We show that lattice functions possess a rich analytic structure in the context of integral-valued entire functions of exponential type. We derive set and spectral properties of this class of functions. This allows us to prove in a deterministic way and for general bandlimited functions a fundamental lower bound on the maximum frequency component introduced by quantization that is independent of the resolution of the quantizer.
\end{abstract}


%
\IEEEpeerreviewmaketitle

\section{Introduction}
\IEEEPARstart{T}{he} acquisition theory of common analog-to-digital (A/D) converters is based on discrete time and analog amplitude. The classical sampling theorem \cite{Whittaker:1915aa,Kotelnikov:1933aa,Shannon:1949aa} provides a theoretical foundation by representing a bandlimited signal by its analog amplitude values at a uniform time grid. Alternatively, it is also possible to consider an acquisition procedure where signal representation is based on quantized amplitude and analog time, i.e. amplitude sampling \cite{Martinez-Nuevo:2016ab}. The latter represents the input signal, for analysis purposes, by the time instants at which a reversible transformation of the input crosses equally-spaced amplitude values. 

In practice, however, the sequences generated by an A/D converter implicitly represent signals that lie between the two sampling paradigms described above, i.e. quantization is present both in the amplitude and time domains. This manifests itself in uniform time sampling, which results in discrete-time processing, and amplitude quantization, i.e. digital representation. Note that it is also possible to perform digital signal processing in continuous time when quantization is introduced \cite{Tsividis:2003aa}.

In this paper, we view A/D conversion from the perspective of sampling theory. We consider the quantized discrete-time sequences generated by an A/D converter as coming from samples of a function that takes amplitude values on a uniform grid at uniform instants of time. We refer to these signals as lattice functions and its acquisition as lattice sampling. In particular, we show that they can be interpreted as integral-valued entire functions of exponential type. 

Based on the analysis of the spectral properties of these functions, we derive a deterministic result about the influence of the quantization error in the spectrum of discrete-time digital sequences. In particular, irrespective of the resolution of the quantizer and the bandwidth of the input signal to an A/D converter, the discrete-time Fourier transform (DTFT) of the digital sequence contains frequency components $\geq0.8$ rad/s. It is common, for theoretical purposes, to consider the set of bandlimited functions in the analysis of A/D conversion. However, in practice, an A/D converter implicitly considers lattice functions. Thus, we derive results about the cardinality of bandlimited lattice functions and examples of functions within this set in order to provide a deeper understanding of the subset of bandlimited signals that an actual A/D converter is implicitly considering.

The approach taken here is in contrast with the traditional analysis of A/D converters based on amplitude quantization as a source of error. The analysis in this context has been mainly addressed from two different perspectives. On the one hand, the probabilistic interpretation considers discrete-time sequences where the quantizer---a nonlinear system---is modeled as a linear system which adds uniformly distributed white noise, in general with a range related to the quantization step \cite{Bennett:1948aa}\cite{Sripad:1977aa}. This noise model is a useful analytic tool used in the analysis of oversampled A/D conversion, and oversampled A/D conversion with noise shaping. 

The deterministic approach has been considered in several scenarios. In the specific case of limit-cycle oscillations in digital filters, it was used to determine bounds on the limit-cycle amplitude for fixed-point implementations of recursive digital filters \cite{Parker:1971aa}. In \cite{Bennett:1948aa}, the spectrum of a quantized sinusoid was also studied in a deterministic way where the well-known result about the output consisting of the odd harmonics of the input frequency was derived. Within the context of digital control systems incorporating quantization, a worst-case upper bound on the dynamic quantization error was derived in \cite{Johnson:1965aa,Lack:1966aa}. In general, the derived bounds are too loose compared with the stochastic approach and experimental data \cite{Parker:1971aa}.

In Section \ref{section:LatticeFunctions}, we formally introduce lattice functions and its relationship to the discrete-time sequences that are the output of an A/D converter. We introduce the latter within the framework of consistent measurements, i.e. consistent resampling and requantization. We show that A/D converters can be viewed as implicitly transforming any input signal to a lattice function when considering an appropriate interpolation of the digital sequence. This suggests that input signals to an A/D converter may be preconditioned to be lattice functions in analogy with antialiasing filtering. 
Section \ref{section:IVBandlimited} studies the properties of the set of functions that a common discrete-time signal processing system deals with due to the transformation performed by an A/D converter. In particular, it shows that lattice functions with a given bandwidth form a countable set within the set of signals with the same bandwidth. Additionally, it gives some examples of a family of functions belonging to this set with the same cardinality. Section \ref{section:IVEntireFunctions} provides results about the connection between integral-valued bandlimited functions and a lower bound on their maximum frequency component. Finally, Section \ref{section:LFSpectralProps} generalizes the previous results to arbitrary lattice functions and provides several implications of this result regarding the Discrete-Time Fourier Transform and interpolation of quantized sequences---with the corresponding connection to the quantization error---, and in connection to Fourier series with quantized coefficients.

Throughout the paper, we consider for a function $f:\mathbb{R}\to\mathbb{C}$ the following definitions of the Fourier transform and its inverse
\begin{equation}
\label{eq:FTx_pairs}
\begin{split}
\hat{f}(\xi)&=\int_{\mathbb{R}}f(t)e^{-i2\pi \xi t}\mathrm{d}t,\ \xi\in\mathbb{R}\\
f(t)&=\int_{\mathbb{R}}\hat{f}(\xi)e^{+i2\pi \xi t}\mathrm{d}\xi,\ t\in\mathbb{R}.
\end{split}
\end{equation}
The time units of the variable $t$ are considered to be seconds. Thus, we denote the set of functions bandlimited to $\sigma$ rad/s by 
\begin{equation}
\mathcal{B}_\sigma=\{f\in L^2(\mathbb{R}):\hat{f}(\xi)= 0\textrm{ for }|\xi|\geq\sigma/2\pi,\ \sigma>0\}.
\end{equation}

The main results regarding the spectrum of integral-valued bandlimited functions are based upon the interpretation of bandlimited signals as entire functions of exponential type. Let us introduce a few concepts that will be useful later. A function complex differentiable at every point in $\mathbb{C}$ is referred to as entire. An entire function $f$ is of exponential type if there exists constants $M,\tau>0$ so that $|f(z)|\leq Me^{\tau|z|}$ for all $z\in\mathbb{C}$. If $\sigma=\inf \tau$ taken over all $\tau$ satisfying the latter inequality, it is said to be of exponential type $\sigma$. A function with appropriate decay conditions $f(t)$ is bandlimited to $\sigma$ rad/s if and only if it can be analytically continued to the whole complex plane in such a way that $|f(z)|\leq e^{\sigma|z|}$ for all $z\in\mathbb{C}$ \cite[Chapter 4, Theorem 3.3]{Stein:2003aa}\cite[Theorem X]{Paley:1934aa}.

\section{Lattice Functions}
\label{section:LatticeFunctions}
A useful approach to studying the class of signals represented by an A/D converter is to first consider a sampling process in which quantization is both present in time and amplitude, we refer to it as \textit{lattice sampling}. This consists of taking samples whenever the source signal crosses exactly the points defined by a two-dimensional grid, i.e. a lattice. In particular, consider the procedure illustrated in Fig.~\ref{fig:LatticeSampling_concept} where the input signals are sampled at the points $\{(nT+\tau,m\Delta+\gamma)\}_{n,m\in\mathbb{Z}}$ where $T$ is the sampling period, $\Delta$ can be interpreted as the quantization step, and $\tau,\gamma\in\mathbb{R}$ represent the offsets. Samples are taken only when the input signal passes through one of these points in a lattice. This approach lies between the classical sampling theorem and amplitude sampling when considered from the point of view of which domains are quantized or assumed to be analog. 

\begin{figure}[!ht]
\centering
\includegraphics[scale=1.2]{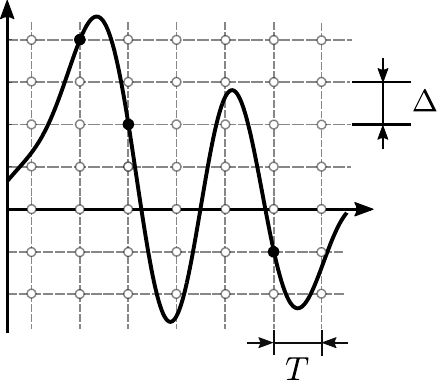}
\caption{Example of sampling in a two-dimensional grid.}
\label{fig:LatticeSampling_concept}
\end{figure}

Although the lattice points are uniformly separated in time and uniformly separated in amplitude, lattice sampling can be nonuniform in time and also in amplitude as illustrated in Fig.~\ref{fig:LatticeSampling_concept}. Our approach throughout this paper is to restrict ourselves to those functions that cross a lattice point at every multiple of $T$, i.e. they take values that are integer multiples of $\Delta$ at every integer multiple of $T$ accounting for the offsets accordingly. We refer to a function satisfying the latter as a \textit{lattice function}.

\begin{definition}
A continuous signal $f:\mathbb{R}\to\mathbb{C}$ is a lattice function if $f(nT+\tau)=m_n\Delta+\gamma$ for all $n\in\mathbb{Z}$ where $m_n\in\mathbb{Z}$, $\Delta,T>0$, and $\gamma,\tau\in\mathbb{R}$.
\end{definition}

\subsection{Consistent Resampling and Requantization}
The connection of lattice functions to A/D converters can be first introduced by considering the block diagram representation in Fig.~\ref{fig:ADC_intro_lattice}. Without loss of generality, consider that $\tau=\gamma=0$. We represent an A/D converter by a continuous-to-discrete (C/D) block followed by a quantizer. The C/D block outputs an analog value $a_n=f(nT)$ every $T>0$ seconds for an input signal $f(t)$ and $n\in\mathbb{Z}$. We consider a uniform quantizer, denoted by $Q_\Delta$ for some $\Delta>0$, with quantization levels $\{m\Delta\}_{m\in\mathbb{Z}}$. Thus, both systems represented in Fig.~\ref{fig:ADC_intro_lattice} are input-output equivalent. 

\begin{figure}[!ht]
\centering
\includegraphics[scale=0.5]{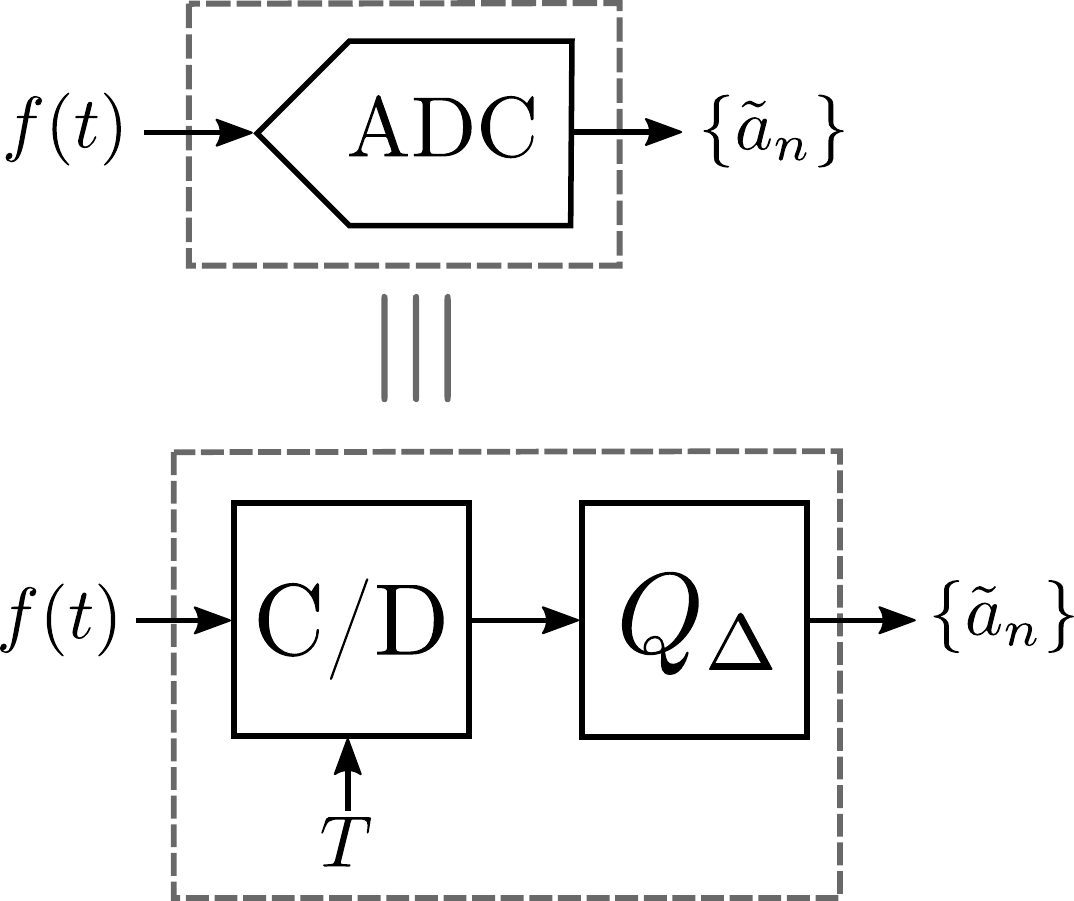}
\caption{Equivalent representation of an ADC consisting of a C/D block that performs uniform time sampling with period $T>0$ followed by a uniform quantizer where $\Delta>0$ is the quantization step.}
\label{fig:ADC_intro_lattice}
\end{figure}

The discrete-to-continuous (D/C) block in Fig.~\ref{fig:ADC_consistent}(a) is an interpolator based on the quantized samples $\{\tilde{a}_n\}$ from a uniform time sampling process at intervals $T$. If the interpolator takes the form 
\begin{equation}
\label{eq:interpolator}
\begin{split}
&\tilde{g}(t)=\sum_{n\in\mathbb{Z}}\tilde{a}_n\psi(t-nT)\\
&\textrm{where }\psi(nT)=0\textrm{ for }n\neq0\textrm{, and }\psi(0)=1,
\end{split}
\end{equation} 
assuming $\psi$ satisfies appropriate decay and regularity conditions, then the output of the D/C block will satisfy $\tilde{g}(nT)=\tilde{a}_n=m_n\Delta$ for $m_n\in\mathbb{Z}$ and all $n\in\mathbb{Z}$. Note that the properties above are also used in digital communications to reduce or eliminate intersymbol interference \cite{Oppenheim:2015aa}. The latter is sufficient for $\tilde{g}$ to be a lattice function. Thus, we can see the A/D converter as implicitly mapping any input signal to a lattice function.

The reconstruction according to (\ref{eq:interpolator}) leads to consistent measurements \cite{Unser:1994aa,Unser:1997aa}. On the one hand, Fig.~\ref{fig:ADC_consistent}(b) shows \textit{consistent resampling}, i.e. $f(nT)=g(nT)=a_n$ where it is not necessary that $f$ and $g$ are identical. On the other hand, the sequence $\{\tilde{a}_n\}$ satisfies \textit{consistent requantization} in the sense that it is not altered when passed through a number of quantizers similar to $Q_\Delta$. Thus, if the input $f$ to the system in Fig.~\ref{fig:ADC_consistent}(a) is a lattice function, then we have that $f(nT)=\tilde{g}(nT)=\tilde{a}_n$ for all $n\in\mathbb{Z}$, i.e. lattice functions lead to \textit{consistent resampling} and \textit{requantization}.

\begin{figure}[!ht]
\centering
\subfloat[]{\includegraphics[scale=0.49]{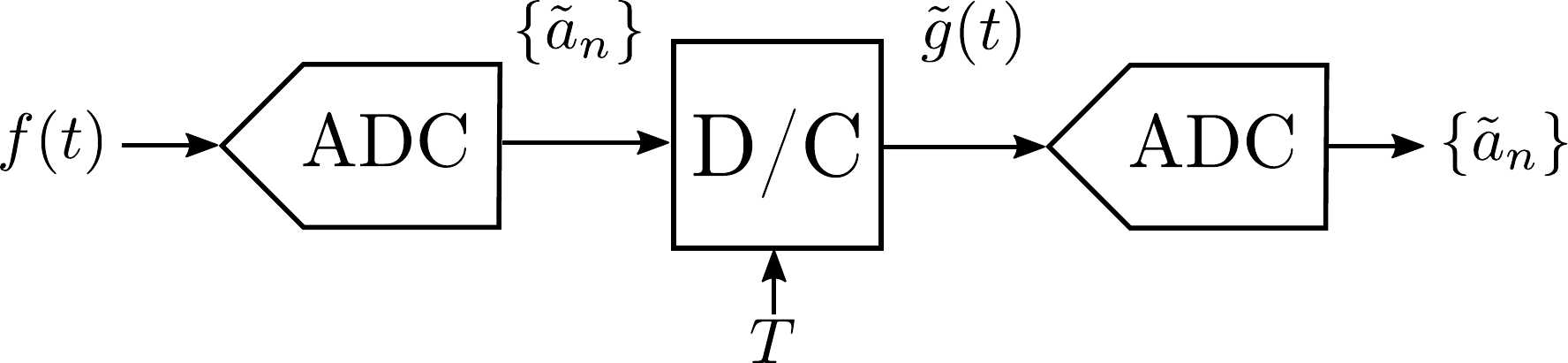}} 
\hfil 
\subfloat[]{\includegraphics[scale=0.5]{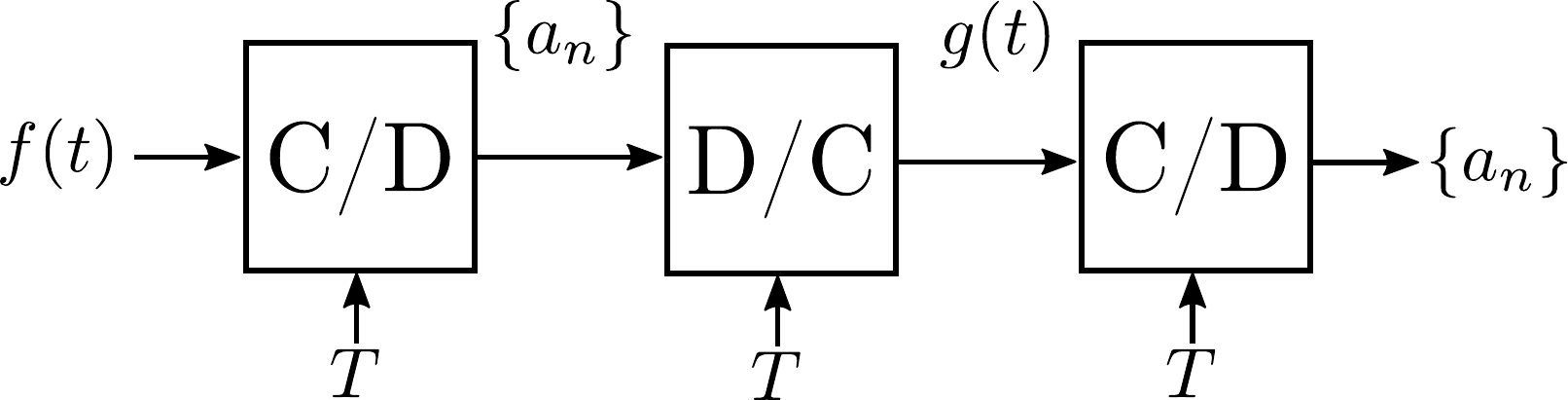}} 
\caption{Consistent measurements assuming the discrete-to-continuous (D/C) block satisfies (\ref{eq:interpolator}). (a) A/D conversion followed by interpolation and A/D conversion showing consistent resampling and requantization, i.e. assuming $f$ is a lattice function, then $g$ is also a lattice function satisfying $f(nT)=\tilde{g}(nT)=\tilde{a}_n$. (b) Consistent resampling, i.e. $f(nT)=g(nT)=a_n$.}
\label{fig:ADC_consistent} 
\end{figure}

\subsection{Integral-valued Lattice Functions}
We have seen that a continuous-time signal to an A/D converter is implicitly mapped to a continuous-time signal that can be sampled at the lattice $\{(mT+\tau,n\Delta+\gamma)\}_{m,n\in\mathbb{Z}}$---i.e. a lattice function---assuming consistent resampling and requantization. For ease of illustration and without loss of generality, we will consider throughout most of the paper the lattice $\{(n,m)\}_{n,m\in\mathbb{Z}}$, i.e. the amplitude levels and the sampling instants correspond, in both cases, to the integers. The generalization to an arbitrary lattice will be carried out in the last section. However, first considering this choice of parameters makes the development and notation less cumbersome. Thus, we will focus on lattice functions that take integer values at the integers, i.e. \textit{integral-valued functions}.

\begin{definition}
A function $f:\mathbb{R}\to\mathbb{C}$ is called integral valued if $f(n)\in\mathbb{Z}$ for all $n\in\mathbb{Z}$.
\end{definition}

The set of lattice functions contains integral-valued functions as a proper subset. It is common to consider in (\ref{eq:interpolator}) an interpolating function of the form $\psi(t)=\mathrm{sinc}((t-nT)/T)$. Thus, $\tilde{g}$ results in a bandlimited lattice function. Our main approach in the study of the spectral characteristics of lattice functions is based on initially assuming bandlimitedness. As a first step towards understanding the mapping performed by A/D converters, the question arises as to whether we can characterize bandlimited integral-valued functions. Due to the connection between polynomials and bandlimited functions, we will first discuss in the next section integral-valued polynomials that will provide insight into the properties of the set of lattice functions.


\section{Integral-Valued Bandlimited Functions in $L^2({\mathbb{R}})$}
\label{section:IVBandlimited}
This section provides insights into the set of functions an A/D converter implicitly converts to, i.e. the underlying bandlimited signals used in discrete-time LTI processing of continuous-time signals. In fact, we will see, by cardinality arguments, that this set of functions is very limited as compared to the common assumption of considering the entire set of bandlimited functions of a specified bandwidth. The latter part of the section illustrates this set of functions by providing an example of a subset with the same cardinality.

Without loss of generality, we first explore the properties of bandlimited lattice functions assuming the lattice points have integer coordinates, i.e. integral-valued bandlimited functions. We initially focus on polynomials. This motivation lies in the fact that continuous bandlimited functions bear a resemblance to polynomials in the sense that they admit a factorization based on their roots in $\mathbb{C}$ due to Hadamard's factorization theorem \cite[Chapter 5]{Stein:2003aa}. Roughly speaking, bandlimited functions can be seen as infinite-degree polynomials and may thus be viewed as a limiting case of finite-degree polynomials. Then, we use this set of integral-valued polynomials to prove the countability of the set of bandlimited lattice functions providing some examples that belong to this set. 

\subsection{Integral-Valued Polynomials}
\label{section:IVPolynomials}
 Our interest lies in polynomials that take integer values at the integers \cite[Chapter 1]{Stanley:2012aa}. In particular, we introduce a characterization of these polynomials in terms of the \textit{difference operator} defined for a function $f:\mathbb{C}\rightarrow\mathbb{C}$ as $\Delta f(n)=f(n+1)-f(n)$ for $n\in\mathbb{Z}$. It can be shown that by applying it $k$ times, we obtain the \textit{$k$-th difference operator}
\begin{equation}
\Delta^{k}f(n)=\Delta(\Delta^{k-1}f(n))=\sum_{i=0}^{k}(-1)^{k-i}\binom{k}{i}f(n+i).
\end{equation}
If we choose $n=0$ we arrive at
\begin{equation}
\label{eq:kdifference0}
\Delta^{k}f(0)=\sum_{i=0}^{k}(-1)^{k-i}\binom{k}{i}f(i)
\end{equation}
which expresses $\Delta^kf(0)$ in terms of the values $f(0),f(1),\ldots,f(k)$. From (\ref{eq:kdifference0}), we can observe that $\Delta^{k}f(0)=0$ for $k>N$ if the function $f$ is a polynomial $P$ of degree $N$. It is also possible to reverse (\ref{eq:kdifference0}) and write $f(n)$ in terms of the finite differences at zero in the following way
\begin{equation}
\label{eq:fnvalues}
f(n)=\sum_{k=0}^{n}\binom{n}{k}\Delta^kf(0)
\end{equation}
for $n\geq0$.

The next result \cite[Corollary 1.9.3]{Stanley:2012aa} provides necessary and sufficient conditions for a polynomial to be integral valued.

\begin{proposition}
\label{prop:intpolys}
Let $P:\mathbb{C}\to\mathbb{C}$ be a polynomial of degree $N$. Then, $P(n)\in\mathbb{Z}$ for all $n\in\mathbb{Z}$ if and only if $\Delta^kP(0)\in\mathbb{Z}$, $0\leq k\leq N$.
\end{proposition}

A sufficient condition for a polynomial to take integer values at the integers is considering integer coefficients. However, the backward assumption in Proposition \ref{prop:intpolys} is more general than the latter. In particular, the coefficients are not required to be integers and, in fact, a closer inspection of (\ref{eq:fnvalues}) reveals that they may belong to $\mathbb{Q}$.

\begin{example}
\normalfont Consider the polynomial $P(x)=x(x-1)/2$ where $\Delta P(0)=0$ and $\Delta^{2}P(0)=1$. By Proposition \ref{prop:intpolys}, the polynomial is integral valued. Indeed, we can use (\ref{eq:fnvalues}) to arrive at the expression
\begin{equation}
P(n)=\sum_{k=0}^{2}\binom{n}{k}\Delta^{k}P(0)=\binom{n}{2}=\frac{n(n-1)}{2}
\end{equation}
where on its right-hand side, we recognize the identity of the sum of the first $n-1$ positive integers. Since for this particular example the polynomial is an even function, we conclude that $P(n)\in\mathbb{Z}$ for $n\in\mathbb{Z}$.
\end{example}

\subsection{Integral-Valued Bandlimited Functions in $L^2({\mathbb{R}})$}
We use the properties of integral-valued polynomials to prove results about the cardinality of lattice functions. We first show that the entire space of functions in $L^2(\mathbb{R})$ and bandlimited to some $\sigma>0$ rad/s cannot be mapped in an unambiguous fashion to a the set of integral-valued polynomials.

\begin{proposition}
\label{prop:PolyBLcountability}
There does not exist a bijection between the set of integral-valued polynomials and the space of functions bandlimited to some $\sigma>0$.
\end{proposition}
\begin{proof}
The space of functions bandlimited to some $\sigma>0$ is a separable Hilbert space, thus it can be identified with $l^2(\mathbb{C})$. By Cantor's diagonal argument \cite[Theorem 2.14]{Rudin:1976aa}, the set of all square-summable sequences is uncountable, thus the aforementioned space of bandlimited functions is uncountable.

By (\ref{eq:fnvalues}) and Proposition \ref{prop:intpolys}, we see that the set of integral-valued polynomials is a subset of the set of polynomials with rational coefficients, i.e. $P(z)=a_0+\ldots+a_Nz^N$ where $a_i\in\mathbb{Q}$ for all $0\leq i\leq N$ and $z\in\mathbb{C}$. We can identify the latter with the set of finite sequences of the form $S=\{(a_0,\ldots,a_N):a_k\in\mathbb{N}\ \mathrm{and}\ N\geq0\}$ noting that $\mathbb{Q}$ is a countable set. Then, a subset $I\subset S$ can be identified with integral-valued polynomials. Consider now the following function
\begin{IEEEeqnarray}{CCCC}
\label{eq:PolyToNat}
\rho:&S&\to&\mathbb{N}\nonumber\\
&(a_0,\ldots,a_N)&\mapsto&p_0^{a_0}\cdot\ldots\cdot p_n^{a_N}-1
\end{IEEEeqnarray}
where $p_i$ is the $(i+1)$-th prime number. The fundamental theorem of arithmetic---i.e. the unique-prime-factorization theorem---implies that the expression in (\ref{eq:PolyToNat}) is a bijection between finite sequences and natural numbers. We can now state that the set $S$ is countable and so is $I$ since $I\subset S$. This implies that the set of integral-valued polynomials is countable. Therefore, there does not exist a bijection between integral-valued polynomials and functions bandlimited to some $\sigma>0$.
\hfill$\square$\end{proof}

In other words, the previous proposition indicates that the \textit{size} of square-integrable bandlimited functions is larger than that of integral-valued polynomials. The previous result was, up to some extent, expected. However, we can intuitively expect that the cardinality of integral-valued bandlimited functions in $L^2(\mathbb{R})$ is the same as integral-valued polynomials. In effect, we show this in the next proposition by demonstrating that the former set is countable.

\begin{proposition}
\label{prop:CountabilityIntBL}
Consider the set 
\begin{equation}
\mathcal{B}_{\pi}^\mathbb{Z}=\{f(t)\in \mathcal{B}_{\pi}:\ f(n)\in\mathbb{Z}\ \textrm{for all}\ n\in\mathbb{Z}\}.
\end{equation}
Then, the set $\mathcal{B}_{\pi}^\mathbb{Z}$ is countable. Moreover, if $f(t)\in \mathcal{B}_{\pi}^\mathbb{Z}$, then $\lim_{|t|\to\infty}f(t)=0$.
\end{proposition}
\begin{proof}
It can be shown that $S:=||f||_2^2=\sum_{n\in\mathbb{Z}}|f(n)|^2<\infty$ where $S>0$ \cite[Chapter 6]{Papoulis:1977ab}. Let us denote the partial sums as $S_N:=\sum_{n=-N}^N|f(n)|^2$. The ring of integers is closed under addition and multiplication, thus $S_N$ is a nonnegative integer. Note that $S_N\rightarrow S$ is a convergent sequence in $\mathbb{R}$, thus it is a Cauchy sequence. This implies that for the particular choice $0<\epsilon<1$, there always exists an integer $N_o>0$ such that $|S_N-S_M|<\epsilon<1$ for all $N,M\geq N_o$. Since $S_N,S_M$ are nonnegative integers, this implies that $S_N=S_M$ for $N,M\geq N_o$ and we have that $f(n)=0$ for $|n|\geq N_o+1$. By the sampling theorem \cite{Whittaker:1915aa,Kotelnikov:1933aa,Shannon:1949aa}, we can then write $f(t)=\sum_{n\in\Lambda}f(n)\mathrm{sinc}(t-n)$ for some finite set $\Lambda\subset\mathbb{Z}$. Then, we can identify $\mathcal{B}_{\pi}^\mathbb{Z}$ with finite sequences of the form $(a_0,\ldots,a_{2N})$ where $a_k\in\mathbb{Z}$ for $0\leq k\leq2N$. As shown in the proof of Proposition \ref{prop:PolyBLcountability}, this set of finite sequences with integer-valued elements is countable and the conclusion follows. As a consequence, since $\mathrm{sinc}(\cdot)\in\mathcal{O}(1/t)$, it follows that $f\in\mathcal{O}(1/t)$ which shows that $\lim_{|t|\to\infty}f(t)=0$.
\hfill$\square$\end{proof}

It is shown in Proposition \ref{prop:CountabilityIntBL} that the functions in $\mathcal{B}_{\pi}^\mathbb{Z}$ possess infinitely many zeros at the integers since they decay at infinity and they are forced to take integer values, or equivalently, they take nonzero values at the integers only at a finite number of times. Additionally, it is clear that $\mathcal{B}_\sigma^{\mathbb{Z}}$ does not form a dense subset within $\mathcal{B}_\sigma$. This would be true for $\mathcal{B}_\sigma^{\mathbb{Q}}$ for example.

We can build upon the idea of finitely many zeros to construct an infinite countable subset of $\mathcal{B}_{\pi}^\mathbb{Z}$ that can help understand the set $\mathcal{B}_{\pi}^\mathbb{Z}$. It is possible to construct a family of these functions without resorting to the canonical series provided by the classical sampling theorem. Consider the sine function divided by an appropriate function to force the decay at infinity according so that the resulting function belongs to $L^2(\mathbb{R})$. At the same time, the resulting function should necessarily be an entire function since we want a bandlimited function. We can choose, for example, polynomials and construct a function like $\sin(\pi t)/P(t)$ where $P(t)$ is a finite-degree polynomial. It is necessary then that the zeros of the sine function cancel those of $P(t)$. Therefore, it is sufficient that this polynomial presents simple zeros at the integers.

In the following result, we show how a family of functions of the form $\sin(\pi t)/P(t)$ are contained in $\mathcal{B}_\pi^{\mathbb{Z}}$ as an infinite countable proper subset. Let us first denote the least common multiple of a finite set $\mathcal{X}$ of natural numbers by $\mathrm{lcm}(\mathcal{X})$.

\begin{proposition}
\label{prop:FormLatticeFunctions}
Let $f$ be a function of the form 
\begin{equation}
\label{eq:BL_IntFuncs}
f(t)=\frac{\sin(\pi t)}{a\prod_{i=1}^{N}(t-t_i)}
\end{equation}
where $N>0$ and $\{t_i\}_{i=1}^{N}\subset\mathbb{Z}$. Then, $f(t)\in \mathcal{B}_\pi$. Moreover, $f(t)$ is integral valued if and only if
\begin{equation}
\label{eq:LCMleadingcoeff}
|a|^{-1}=k\cdot\mathrm{lcm}(\{\prod_{i\neq j}|t_i-t_j|:1\leq j\leq N\})
\end{equation}
for some integer $k>0$.
\end{proposition}
\begin{proof}
The polynomial in the denominator has the form
\begin{equation}
P(t)=\prod_{i=1}^{N}(t-t_i)=\sum_{n=0}^Na_nt^n
\end{equation}
where $t_i\in\mathbb{Z}$ for $1\leq i\leq N$ and $t_i\neq t_j$ for $i\neq j$. It is clear that the function $f(z)$ has removable singularities at the zeros of $P(z)$, thus it is an entire function. 

Note that there exist a $C>0$  and $t_o>0$ such that $|P(t)|\geq C|t|$ for $t>t_o$. Then, we have the following
\begin{equation}
\int_{\mathbb{R}}|f(t)|^2\mathrm{d}t\leq\int_{|t|>t_o}\Big|\frac{\sin(\pi t)}{\pi Ct}\Big|^2\mathrm{d}t+D<\infty
\end{equation}
for some $0<D<\infty$. Then, it is immediate to see that $f(t)\in L^2(\mathbb{R})$.

Let us denote the type of $f$ by $\sigma$. The function $\sin(\pi z)$ is of type $\pi$, thus using Euler's identity we have that $\sigma\leq\pi$.
It is straightforward to see that for every $0<\delta<\pi$ and $C'>0$ there exist a large enough $z_o\in\mathbb{C}$ such that
\begin{equation}
|f(z_o)|\geq \frac{Ce^{\pi |z_o|}}{A|z_o|^{n}}>C'e^{(\pi-\delta)|z_o|}
\end{equation}
where $A=\sum_{n=0}^N|a_n|$. This implies that $\sigma\geq\pi$, and consequently, $\sigma=\pi$. By the Paley-Wiener Theorem \cite[Chapter 4, Theorem 3.3]{Stein:2003aa}\cite[Theorem X]{Paley:1934aa}, the function $f(t)\in L^2(\mathbb{R})\cap\mathcal{C}^{0}(\mathbb{R})$ is bandlimited to $[-\pi,\pi]$.

The function $f$ has removable singularities at $\{t_i\}_{i=1}^{N}$, thus we can define the values at these points as
\begin{equation}
\label{eq:PointDefSincLike}
f(t_i)=\lim_{t\to t_i}f(t)=\frac{\cos(\pi t_i)}{P'(t_i)}=\frac{(-1)^{t_i}}{P'(t_i)}
\end{equation}
where we have applied L'H\^opital's rule. Note that (\ref{eq:PointDefSincLike}) is well defined as $P'(t_i)\neq0$ for all $1\leq i\leq N$. In order to see this, consider some $1\leq j\leq N$ and write the derivative of the polynomial as
\begin{equation}
\label{eq:DerPoly}
P'(t)=\prod_{i\neq j}(t-t_i)+(t-t_j)\frac{\mathrm{d}}{\mathrm{d}t}\prod_{i\neq j}(t-t_i).
\end{equation}
Since the polynomial has simple roots, the first term in (\ref{eq:DerPoly}) will be different from zero for $t=t_j$ and the second term will vanish for all $1\leq j\leq N$. In view of the preceding, we can also write
\begin{equation}
f(t_i)=\frac{(-1)^{t_i}}{a\prod_{j\neq i}(t_i-t_j)}.
\end{equation}
Note that the product $\prod_{j\neq i}(t_i-t_j)$ is an integer for all $1\leq j\leq N$. In fact, it clear that it is necessary and sufficient that $|a|^{-1}$ is a multiple of the values $\{\prod_{i\neq j}|t_i-t_j|:1\leq i,j\leq N\}$ in order for $f$ to take integral values at the integers.
\hfill$\square$\end{proof}

We have seen that if we appropriately choose the roots of $P(t)$ and its leading coefficient $a$, the function $\sin(\pi t)/P(t)$ is an integral-valued bandlimited function in $L^2(\mathbb{R})$, i.e. a square-integrable lattice function for $\Delta=T=1$. Denote this family of functions by
\begin{equation}
\begin{split}
\mathcal{S}=\Big\{\frac{\sin(\pi t)}{\pi a\prod_{i=1}^{N}(t-t_i)}:\ &\mathrm{distinct}\ t_i\in\mathbb{Z}, N>0\\
& \mathrm{and}\ a\ \mathrm{satisfying\ (\ref{eq:LCMleadingcoeff})}\Big\}
\end{split}
\end{equation}
that, based on the previous result, satisfies $\mathcal{S}\subseteq \mathcal{B}_{\pi}^{\mathbb{Z}}$. Note again that the generalization to an arbitrary quantizer step and sampling rate can be performed by an appropriate time warping and scaling. The set $\mathcal{S}$ in effect forms a countable subset of square-integrable lattice functions as the next result shows.

\begin{corollary}
There exists a bijection between the set $\mathcal{B}_{\pi}^{\mathbb{Z}}$ and $\mathcal{S}$.
\end{corollary}
\begin{proof}
The set $\mathcal{S}$ can be identified with the set of sequences of the form $(k,t_1,t_2,\ldots,t_N)$ where $k$ is some positive integer as in (\ref{eq:LCMleadingcoeff}). By the same argument presented in Proposition \ref{prop:CountabilityIntBL}, $\mathcal{S}$ is countable. Since $\mathcal{B}_{\pi}^{\mathbb{Z}}$ is also countable, both sets can be identified, i.e. there exist a bijection between them.
\hfill$\square$\end{proof}

Roughly speaking, the \textit{size} of $\mathcal{S}$ is the same as $\mathcal{B}_{\pi}^{\mathbb{Z}}$, however, $\mathcal{S}$ lives within $\mathcal{B}_{\pi}^{\mathbb{Z}}$ as a proper subset, i.e. $\mathcal{S}\subset \mathcal{B}_{\pi}^{\mathbb{Z}}$ where the inclusion is strict. In order to see this, let us illustrate it in the following example.

\begin{example}
\label{ex:ExBLintSinc}
\normalfont Assume we have the functions
\begin{equation}
\label{eq:ExBLintSinc}
\begin{split}
g(t) = &\ 6\mathrm{sinc}(\pi(t+1))+5\mathrm{sinc}(\pi(t-2))+\\
& +\mathrm{sinc}(\pi(t+3))\\
f(t) = &\ \frac{\sin(\pi t)}{\pi(t+1)(t-2)(t+3)}
\end{split}
\end{equation}

where $f(t)\in\mathcal{S}$ and $g(t)\in \mathcal{B}_{\pi}^{\mathbb{Z}}$. Note that the function $f(t)$ satisfies the conditions imposed by Proposition \ref{prop:FormLatticeFunctions}, in fact, for any $k\in\mathbb{Z}$ we have that $kf(t)\in\mathcal{S}$. We can easily see in Fig.~\ref{fig:ExBLintSinc} that there does not exist such $k$ in such a way that both functions are equal for all $t\in\mathbb{R}$. Thus, this counter example clearly implies that $\mathcal{S}$ is a proper subset of $ \mathcal{B}_{\pi}^{\mathbb{Z}}$.
\end{example}

Proposition \ref{prop:FormLatticeFunctions} allows us to show in the next result an identity---proven from a sampling-theoretic point of view---involving functions of the form (\ref{eq:BL_IntFuncs}) and finite linear combinations of sinc functions when the coefficients are appropriately chosen.

\begin{figure}[!t]
\centering
\includegraphics[scale=0.45]{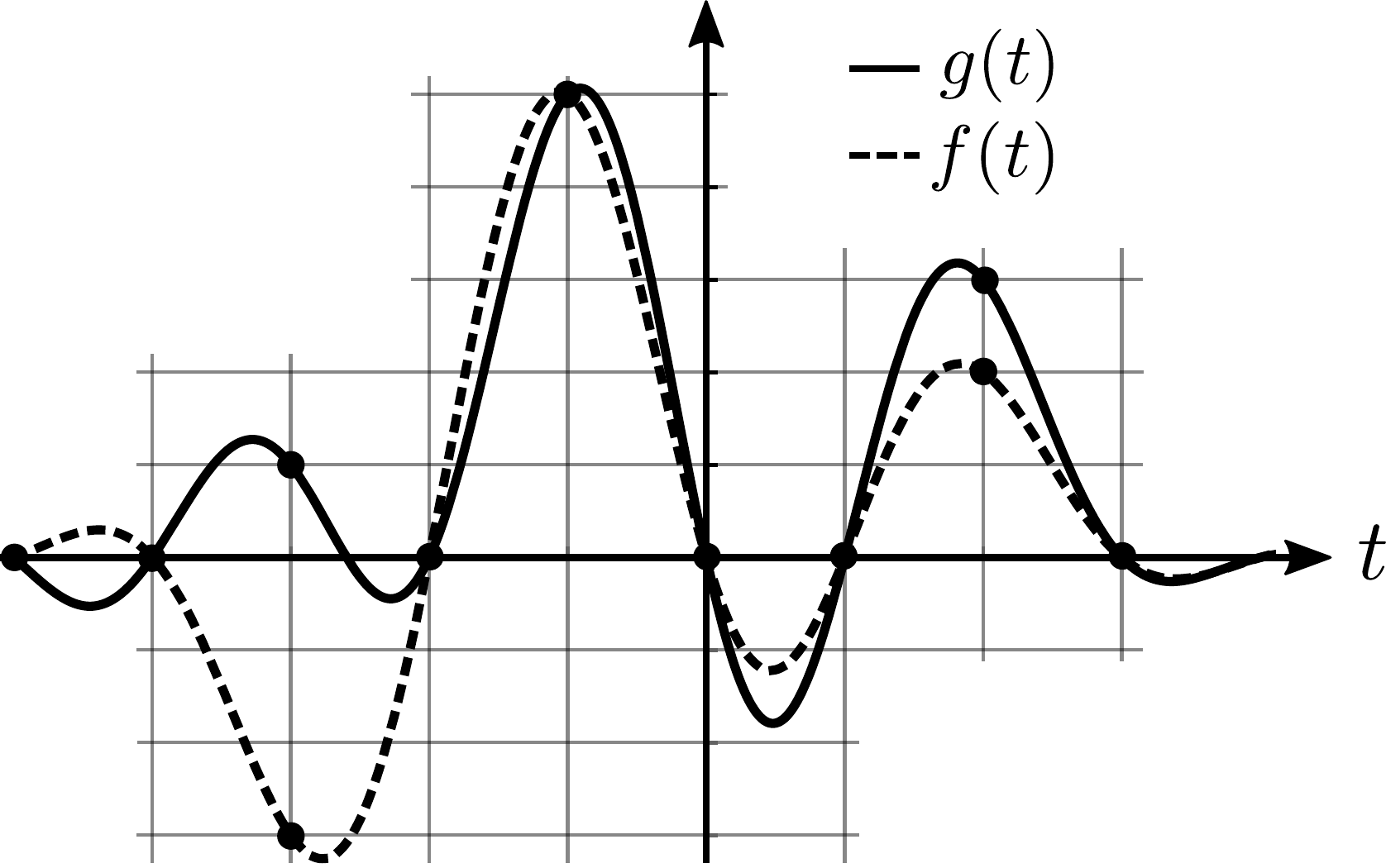}
\caption{Plot of the functions of Example \ref{ex:ExBLintSinc} where $f(t)\in\mathcal{S}$ and $g(t)\in\mathcal{B}_{\pi}^{\mathbb{Z}}$. The black dots represent crossings in the lattice grid formed at the integers in both axis, i.e. $\Delta=T=1$.}
\label{fig:ExBLintSinc}
\end{figure}

\begin{corollary}
Assume $P(t)=a\prod_{i=1}^{N}(t-t_i)$ where $|a|^{-1}$ satisfies (\ref{eq:LCMleadingcoeff}) and $\{t_i\}_{i=1}^{N}\subset\mathbb{Z}$. If the values $\{\gamma_i\}_{i=1}^{N}$ are chosen such that there exists a nonzero integer $k$ satisfying $\gamma_j=k(-1)^{t_j}/a\prod_{i\neq j}(t_j-t_i)$ for all $1\leq j\leq N$, then the following identity holds
\begin{equation}
\label{eq:SincIdentity}
\sum_{i=1}^{N}\gamma_i\mathrm{sinc}(t-t_i)=\frac{\sin(\pi t)}{k\pi P(t)}.
\end{equation}
\end{corollary}
\begin{proof}
Both sides of expression (\ref{eq:SincIdentity}) vanish at $\mathbb{Z}\setminus\{t_i\}_{i=1}^{N}$. Since $a$ is chosen such that (\ref{eq:LCMleadingcoeff}) is satisfied, the values of the right-hand side of the previous expression at the roots of $P(t)$ are precisely $(-1)^{t_j}/a\prod_{i\neq j}(t_j-t_i)$ for $1\leq j\leq N$. The latter implies that the expression is valid for the integers. Moreover, since both sides are bandlimited to $[-\pi,\pi]$, by the sampling theorem \cite{Whittaker:1915aa,Kotelnikov:1933aa,Shannon:1949aa}, they agree for all $t\in\mathbb{R}$.
\hfill$\square$\end{proof}

\section{Integral-Valued Entire Functions}
\label{section:IVEntireFunctions}
As introduced above, the main results regarding the spectrum of integral-valued bandlimited functions are based upon the interpretation of bandlimited signals as entire functions of exponential type. Loosely speaking, bandlimited signals correspond to entire functions with an exponential growth on the whole complex plane. Moreover, the growth determines the bandwidth of the signal itself, or vice versa. Thus, we explore the growth of integral-valued entire functions that will as a result reveal the spectral properties of bandlimited lattice functions.

The first results that connect the rate of growth of integral-valued entire functions with their structure can be found in \cite{Polya:1915aa,Polya:1920aa,Hardy:1917aa}. In particular, it can be shown that if an entire function takes integer values for nonnegative integers and $\sigma<\log(2)$, this function has to be a polynomial \cite[Theorem 11]{Whittaker:1935aa}. This result was extended and refined in many instances in the literature \cite{Selberg:1941aa,Selberg:1941ab,Pisot:1942aa,Pisot:1946aa,Pisot:1946ab,Buck:1948aa}. We now state the result in \cite{Boas:1954aa} that we will be using later in connection with bandlimited signals.

\begin{theorem}
\label{thm:IntEFET}
Assume that $f(z)$ is an entire function of exponential type $\sigma$, with $f(n)\in\mathbb{Z}$ for integers $n\geq0$. If the type satisfies
\begin{equation}
\sigma<\Big|\log\Big(\frac{3}{2}+\frac{i\sqrt{3}}{2}\Big)\Big|=0.7588\ldots
\end{equation}
then $f(z)$ is of the form $P_0(z)+P_1(z)2^z$, where $P_0$ and $P_1$ are polynomials. Moreover, if $\sigma<0.8$, then
\begin{equation}
\label{eq:IntExpression}
\begin{split}
f(z) =\ &P_0(z)+P_1(z)2^z+P_2(z)\Big(\frac{3}{2}+\frac{i\sqrt{3}}{2}\Big)^z+\\
&+P_3(z)\Big(\frac{3}{2}-\frac{i\sqrt{3}}{2}\Big)^z.
\end{split}
\end{equation}
\end{theorem}

Essentially, these functions take the shape of a sum of finitely many terms of the form $a^zP(z)$, where $a$ is an \textit{algebraic integer} (a complex number which is a root of a monic polynomial with integer coefficients). Notice that it is not possible to increase the type up to $\sigma=\pi$ since $\sin(\pi z)$ is an integral-valued entire function of exponential type $\pi$.

Theorem \ref{thm:IntEFET} relates the type and the structure of an integral-valued entire function. However, we are interested in integral-valued bandlimited signals that are square integrable, thus we can use the previous theorem to derive a result regarding the bandwidth of such signals.

\begin{proposition}
\label{prop:entireBW}
If $f(t)\in L^2(\mathbb{R})$ is a bandlimited function with $f(n)\in\mathbb{Z}$ for all integers $n\geq0$, then the type $\sigma$ of $f(z)$ satisfies $\sigma\geq0.8$.
\end{proposition}
\begin{proof}
By the Paley-Wiener theorem, we know that $f(t)$ admits an analytic continuation as an entire function of exponential type $\sigma$. However, the function $f$ is square-integrable in the real line, thus it cannot be expressed in the form described in (\ref{eq:IntExpression}). Therefore, by the contrapositive of Theorem \ref{thm:IntEFET}, we have that $\sigma\geq0.8$.
\hfill$\square$\end{proof}

Due to the connection between the bandwidth of a square-integrable bandlimited signal and the type of its analytic extension, we can interpret the previous result as providing as a lower bound for the maximum frequency component of integral-valued bandlimited signals. Since these signals satisfy $f(n)\in\mathbb{Z}$ for all $n\in\mathbb{Z}$, Proposition \ref{prop:entireBW} implies that their maximum frequency component has to be always greater or equal than $0.8$ rad/s. In summary, there are no square-integrable bandlimited signals that take integer values at integer points whose spectrum is confined to the interval $(-0.8,0.8)$ rad/s.

In the following, we show how to utilize this result to construct functions with a certain bandwidth by choosing its sample values. In the next section, we will extend this result to lattice functions and will make the connection to the quantization error introduced by A/D converters.

\subsection{Constructing functions with bandwidth between $0.8$ and multiples of $\pi$.}
Proposition \ref{prop:entireBW} establishes a relationship between sample values, in this case integers, and the bandwidth of a signal. We now show how to construct a function whose maximum frequency component lies in an interval based solely on appropriately choosing a subset of its sample values as integers. We make this notion precise in the following result.
\begin{corollary}
\label{cor:BWconstruction}
If $\{a_n\}_{n\in\mathbb{Z}}$ is a complex square-summable sequence such that $a_n\in\mathbb{Z}$ for all $n\geq0$, then the type $\sigma$ of the entire function constructed as
\begin{equation}
\label{eq:canonicalseries}
f(t)=\sum_{n\in\mathbb{Z}}a_n\mathrm{sinc}(t-n)
\end{equation}
satisfies $0.8\leq\sigma\leq\pi$. Moreover, there are finitely many $a_n\neq0$ for $n\geq0$.
\end{corollary}
\begin{proof}
Since the sequence is square-summable, by Parseval's identity we have that $f(t)\in L^2(\mathbb{R})$. The latter implies that $f(t)$ cannot be a sum with terms of the form $a^zP(z)$. According to Lemma \ref{prop:entireBW}, the function $f(z)$ is of type $\sigma\geq0.8$. By construction, we also have the upper bound $\sigma\leq\pi$.
The square-summability of the sequence implies that $a_n\to0$ as $n\to\infty$. This implies that there exists an $N>0$ such that $a_n=0$ for all $n\geq N$ and the conclusion follows.
\hfill$\square$\end{proof}

The previous result shows a way of constructing a function whose maximum frequency component is between $0.8$ and $\pi$ rad/s. We only have to construct a square-summable sequence with any value for negative integers and finitely many nonzero integer values for $n\geq0$ setting the remaining ones to zero. Notice that it is not guaranteed that any square-summable sequence in (\ref{eq:canonicalseries}) generates a function whose bandwidth is $\pi$. Indeed, a straightforward counterexample is oversampling. If we take the values of a sinewave at the integers whose frequency is strictly smaller than $\pi/2$ rad/s, then the sequence $\{f(n)\}_{n\in\mathbb{Z}}$ together with (\ref{eq:canonicalseries}) will result in a function whose bandwidth is strictly smaller than $\pi/2$.

Interestingly, we can easily extend the procedure and consider as interpolating functions $\{\mathrm{sinc}(t/k-n)\}_{n\in\mathbb{Z}}$ for some integer $k\geq1$. In this case, we can choose $a_{nk}\in\mathbb{Z}$ for integers $n\geq0$. The latter are the sample values that correspond to integer sampling points in the sampling sequence. In this situation, we ensure that the bandwidth of the resulting function is between $0.8$ and $k\pi$. In other words, we decimate by a factor of $k$ the sequence for nonnegative integers and restrict those to be integers.

\section{Spectral Properties of Lattice functions}
\label{section:LFSpectralProps}
We develop in this section the main implications of the result of Proposition \ref{prop:entireBW}. In particular, we start by framing the result from the perspective of the Discrete-Time Fourier Transform (DTFT) of quantized sequences where we present a lower bound on the maximum frequency component of quantized sequences---this can be seen as an effect of the quantization error. Then, we discuss the spectral properties of the interpolation of quantized sequences, and finally we add a result concerning the Fourier series coefficients of periodic signals.

The previous sections focused on integral-valued bandlimited signals to cleanly introduce the results. We now generalize these results to arbitrary bandlimited lattice functions. Assume that the lattice points are $\{(nT+\tau,m\Delta+\gamma)\}_{n,m\in\mathbb{Z}}$ for $T,\Delta>0$ and $\tau,\gamma\in\mathbb{R}$. The lattice functions in this case take values $m\Delta+\gamma$ for some $m\in\mathbb{Z}$ at every instant $nT+\tau$ for $n\in\mathbb{Z}$. Assume further that these functions are bandlimited with appropriate decay conditions on the real line. It is now possible to relate their properties to those of integral-valued bandlimited functions. In order to do so, consider the lattice function $f(t)$ for the lattice described above and construct the function
\begin{equation}
f_o(t)=:\frac{1}{\Delta}(f(tT+\tau)-\gamma)
\end{equation}
where it follows that $f_o(\cdot)$ is an integral-valued bandlimited function. The Fourier transforms are then related by
\begin{equation}
\label{eq:FTxscaling}
\hat{f}(\xi)=T\Delta e^{-i2\pi\xi\tau}\hat{f}_o(T\xi)+\gamma\delta(\xi)
\end{equation}
where $\delta(\cdot)$ is the Dirac delta function. In view of (\ref{eq:FTxscaling}), the results in the previous sections can be easily extended to general lattice functions by scaling the bandwidth by a factor of $1/T$. In particular, the bound in Proposition \ref{prop:entireBW} will now be $\sigma\geq 0.8/T$ rad/s. Similarly, the bounds derived in Corollary \ref{cor:BWconstruction} will be of the form $0.8/T\leq\sigma\leq k\pi/T$ rad/s.

Thus, any bandlimited lattice function on a grid $\{(nT+\tau,m\Delta+\gamma)\}_{n,m\in\mathbb{Z}}$ has spectral content at least up to $0.8/T$ rad/s irrespective of the value of $\Delta$ or $\gamma$, i.e. irrespective of the resolution of the quantizer.


\subsection{DTFT of Quantized Discrete-Time Signals}
The previous results state that any bandlimited lattice function for the lattice points $\{(nT+\tau,m\Delta+\gamma)\}_{n,m\in\mathbb{Z}}$ has frequency components that extend at least up to $0.8/T$ rad/s. Since the DTFT of a discrete-time sequence $x[n]$ is related to the Fourier transform of its bandlimited interpolation, i.e. $x(t)=\sum_{n\in\mathbb{Z}}x[n]\mathrm{sinc}((t-nT)/T)$, we can formally express in the next result that the maximum frequency component of the DTFT of quantized sequences lies between $0.8$ and $\pi$ rad/s.


\begin{corollary}
\label{cor:DTFTquant}
Let $\tilde{x}[n]\in\ell^2(\mathbb{C})$ be a sequence such that $\tilde{x}[n]=m_n\Delta+\gamma$ for $m_n\in\mathbb{Z}$, $\Delta>0$, $\gamma\in\mathbb{R}$, and all $n\in\mathbb{Z}$. Then, the Discrete-Time Fourier transform of $\tilde{x}[n]$ given by
\begin{equation}
\tilde{X}(e^{i\omega})=\sum_{n\in\mathbb{Z}}\tilde{x}[n]e^{-i\omega n}
\end{equation}
for all $\omega\in\mathbb{R}$, satisfies $\tilde{X}(e^{j\omega})\neq0$ for $\omega$ in some nonempty interval contained in $[0.8,\pi]$.
\end{corollary}

It should be emphasized that the previous result is independent of the resolution of the quantizer, and it only depends on the fact that the sequence has been quantized. Fig.~\ref{fig:ADC_QSeq} illustrates how any analog signal---with appropriate decay conditions---passed through an A/D converter produces a quantized sequence whose bandwidth lies in the interval described above. In other words, it is not possible to construct a quantized square-summable sequence with a bandwidth, in the sense of the DTFT, smaller than 0.8 rad/s. This also provides a fundamental limit for the modification of the bandwidth by the quantization. Assume the DTFT of a sequence following the C/D block is only nonzero in an interval $(-\omega_o,\omega_o)$ rad/s where $\omega_o<0.8$ rad/s. After the quantizer, the sequence must have frequency components greater or equal than $0.8$ rad/s irrespective of the resolution $\Delta$.

\begin{figure}[!t]
\centering
\includegraphics[scale=0.36]{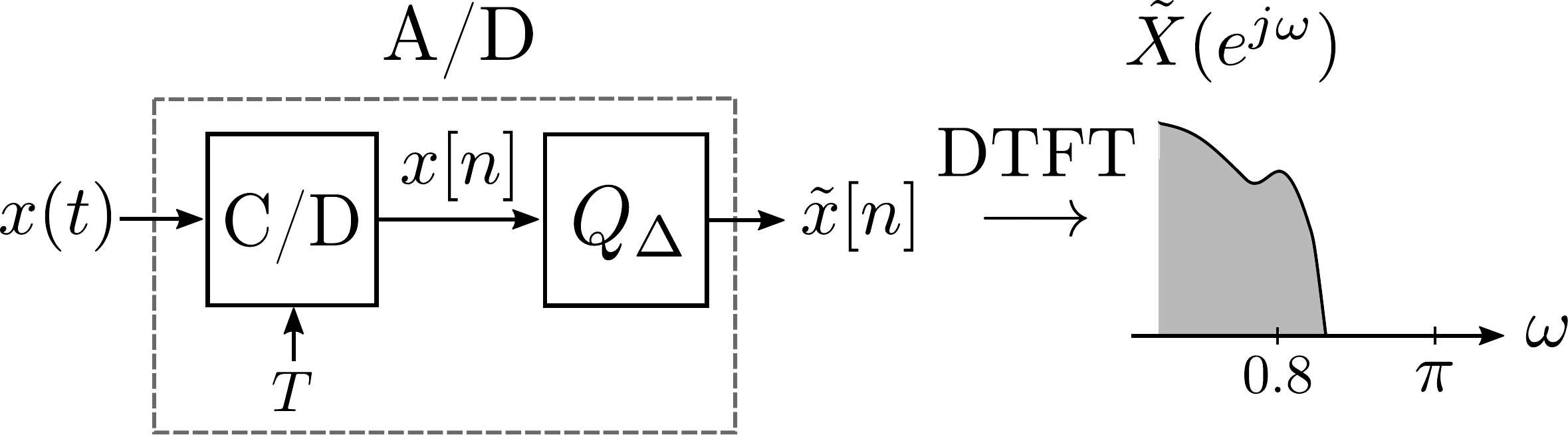}
\caption{Illustration of the spectral properties of a sequence generated by an A/D converter. The input signal is assumed to have the appropriate decay conditions in order to produce a square-summable sequence.}
\label{fig:ADC_QSeq}
\end{figure}

\begin{example}
For analysis purposes, sampling and quantization can be interchanged without altering the result. Thus, it is well known that a sinusoid with frequency $\Omega_o$ passed through a symmetric quantizer contains components at odd harmonics $k\Omega_o$ for $k=\pm1,\pm3,\pm5,\ldots$. Sampling causes these harmonics to be aliased at $k\Omega_o+l\Omega_s$ for $l\in\mathbb{Z}$ and a sampling frequency $\Omega_s>2\Omega_o$. Thus, it is clear that this quantization noise, which here takes the form of aliased harmonic distortion, is in agreement with Corollary \ref{cor:DTFTquant}.
\end{example}

\subsection{Interpolation of Quantized Sequences}
As shown above, quantized sequences have a particular bandwidth characteristic in terms of the DTFT. This property can have an impact on the bandwidth of any interpolation performed based on these samples. Proposition \ref{prop:entireBW} and Corollary \ref{cor:DTFTquant} can then be interpreted in the following manner. Take a quantized sequence, interpolate it to a continuous function according to (\ref{eq:interpolator}). The resulting continuous function will be a lattice function and will always have a bandwidth $\geq0.8/T$ rad/s. If the interpolating function corresponds to a sinc---i.e. $\psi(\cdot)=\mathrm{sinc}(\cdot/T)$---, then the bandwidth will be between $0.8/T$ and $\pi/T$ rad/s.

\begin{example}
\label{ex:inputBWconstraint}
Let us consider the block diagram depicted in Fig.~\ref{fig:ADC_QSeq}. Assume the input signal $x$ has a maximum frequency component satisfying $\Omega_o<0.8/T$ rad/s. By Proposition \ref{prop:entireBW}, this means that $x$ is not a lattice function. Moreover, as shown before, Corollary \ref{cor:DTFTquant} states that $\tilde{x}[n]$ has a maximum frequency component $\geq$0.8 rad/s. Thus, any interpolation of this sequence of the form (\ref{eq:interpolator}), i.e. $\tilde{g}(t)=\sum_{n\in\mathbb{Z}}\tilde{x}[n]\psi(t-nT)$, will always result in a lattice function, thus having a bandwidth above $0.8/T$ rad/s.

\end{example}

\subsection{Quantized Fourier Series Coefficients}
It is possible to relate the Fourier series coefficients of a square-integrable function $\hat{u}$ supported on $[-\xi_o/2,\xi_o/2]$ for some $\xi_o>0$ to the samples of its inverse Fourier transform $u$, i.e. $\{u(n/\xi_o)\}$. In particular, the Fourier series of $\hat{u}$ is given by
\begin{equation}
\label{eq:Fourierseries}
\hat{u}(\xi)=\sum_{n\in\mathbb{Z}}c_ne^{+i2\pi n\xi/\xi_o}
\end{equation}
for $\xi\in[-\xi_o/2,\xi_o/2]$ where
\begin{equation}
c_{-n}:=\frac{1}{\xi_o}\int_{-\xi_o/2}^{+\xi_o/2}\hat{u}(\xi)e^{+i2\pi\xi n/\xi_o}\mathrm{d}\xi=\frac{1}{\xi_o}u(\frac{n}{\xi_o}).
\end{equation}
Note that the Fourier series coefficients are given, up to a scaling factor, by the samples of the inverse Fourier transform of $\hat{u}$. By Proposition \ref{prop:entireBW}, we can relate the support of $\hat{u}$ to the values $\{c_n\}$ (in order to see this, it can be useful to interpret $\hat{u}$ as the Fourier transform of a time-domain signal $u$). 

Similar to Example \ref{ex:inputBWconstraint}, the Fourier series expansion given by (\ref{eq:Fourierseries}) of a function $\hat{u}$ supported on a compact interval strictly contained in $[0.8\xi_o/2\pi,0.8\xi_o/2\pi]$ cannot have coefficients satisfying $c_n=(1/\xi_o)(m_n\Delta+\gamma)$ for $m_n\in\mathbb{Z}$ and all $n\in\mathbb{Z}$ where $\Delta>0$ and $\gamma\in\mathbb{R}$, i.e. quantized Fourier series coefficients. Note, that it is possible that some of the coefficients satisfy the latter, but not all of them. If they were, that would mean that $u$ is a lattice function with a corresponding Fourier transform having nonzero frequency components above $0.8\xi_o$ rad/s, hence a contradiction. This means, for example, that if $\{c_n\}$ is a square-summable sequence of integers, it cannot represent in the form of (\ref{eq:Fourierseries}) a function $\hat{u}$ supported on a compact  interval strictly contained in $[-0.8/2\pi,0.8/2\pi]$.


\section{Conclusion}
We presented a deterministic theoretical analysis of the signals that common A/D converters output, i.e. digital discrete-time signals. We interpreted the interpolation of these sequences with the concept of consistent resampling and requantization. We placed the bandlimited interpolation of these signals---i.e. bandlimited lattice functions---within the framework of integral-valued entire functions to analyze its set and spectral properties. We showed their structure within the space of bandlimited functions and proved a lower bound on their maximum frequency component. This allows to interpret the influence of the quantization error in the spectrum of quantized discrete-time signals. The work shown here suggests that viewing digital discrete-time signals as integral-valued entire functions may provide a theoretical framework where robust deterministic analysis of quantization effects can be performed.

\ifCLASSOPTIONcaptionsoff
  \newpage
\fi

\bibliographystyle{IEEEtran}
\bibliography{/Users/pmnuevo/Documents/BibDeskLibrary/BibDeskLibrary}

\begin{thebibliography}{10}
\providecommand{\url}[1]{#1}
\csname url@samestyle\endcsname
\providecommand{\newblock}{\relax}
\providecommand{\bibinfo}[2]{#2}
\providecommand{\BIBentrySTDinterwordspacing}{\spaceskip=0pt\relax}
\providecommand{\BIBentryALTinterwordstretchfactor}{4}
\providecommand{\BIBentryALTinterwordspacing}{\spaceskip=\fontdimen2\font plus
\BIBentryALTinterwordstretchfactor\fontdimen3\font minus
  \fontdimen4\font\relax}
\providecommand{\BIBforeignlanguage}[2]{{%
\expandafter\ifx\csname l@#1\endcsname\relax
\typeout{** WARNING: IEEEtran.bst: No hyphenation pattern has been}%
\typeout{** loaded for the language `#1'. Using the pattern for}%
\typeout{** the default language instead.}%
\else
\language=\csname l@#1\endcsname
\fi
#2}}
\providecommand{\BIBdecl}{\relax}
\BIBdecl

\bibitem{Whittaker:1915aa}
E.~T. Whittaker, ``{XVIII}.---{O}n the functions which are represented by the
  expansions of the interpolation-theory,'' \emph{Proceedings of the Royal
  Society of Edinburgh}, vol.~35, pp. 181--194, 1915.

\bibitem{Kotelnikov:1933aa}
V.~A. Kotelnikov, ``On the carrying capacity of the ether and wire in
  telecommunications,'' in \emph{Material for the First All-Union Conference on
  Questions of Communication, Izd. Red. Upr. Svyazi RKKA, Moscow}, 1933.

\bibitem{Shannon:1949aa}
C.~E. Shannon, ``Communication in the presence of noise,'' \emph{Proceedings of
  the IRE}, vol.~37, no.~1, pp. 10--21, 1949.

\bibitem{Martinez-Nuevo:2016ab}
P.~Martinez-Nuevo, H.~Lai, and A.~V. Oppenheim, ``Amplitude sampling.''\hskip
  1em plus 0.5em minus 0.4em\relax Allerton Conference on Communication,
  Control and Computing, 2016.

\bibitem{Tsividis:2003aa}
Y.~Tsividis, ``Continuous-time digital signal processing,'' \emph{Electron.
  Lett.}, vol.~39, no.~21, pp. 1551--1552, 2003.

\bibitem{Bennett:1948aa}
W.~R. Bennett, ``Spectra of quantized signals,'' \emph{Bell Labs Technical
  Journal}, vol.~27, no.~3, pp. 446--472, 1948.

\bibitem{Sripad:1977aa}
A.~Sripad and D.~Snyder, ``A necessary and sufficient condition for
  quantization errors to be uniform and white,'' \emph{IEEE Transactions on
  Acoustics, Speech, and Signal Processing}, vol.~25, no.~5, pp. 442--448,
  1977.

\bibitem{Parker:1971aa}
S.~Parker and S.~Hess, ``Limit-cycle oscillations in digital filters,''
  \emph{IEEE Transactions on Circuit Theory}, vol.~18, no.~6, pp. 687--697,
  1971.

\bibitem{Johnson:1965aa}
G.~Johnson, ``Upper bound on dynamic quantization error in digital control
  systems via the direct method of liapunov,'' \emph{IEEE Transactions on
  Automatic Control}, vol.~10, no.~4, pp. 439--448, 1965.

\bibitem{Lack:1966aa}
G.~N.~T. Lack, ``Comments on 'upper bound on dynamic quantization error in
  digital control systems via direct method of lyapunov','' \emph{IEEE
  Transactions on Automatic Control (Corresp.)}, vol. AC-11, no.~2, pp.
  331--333, April 1966.

\bibitem{Stein:2003aa}
E.~M. Stein and R.~Shakarchi, \emph{Complex analysis}, ser. Princeton Lectures
  in Analysis, II.\hskip 1em plus 0.5em minus 0.4em\relax Princeton University
  Press, Princeton, NJ, 2003.

\bibitem{Paley:1934aa}
R.~E. A.~C. Paley and N.~Wiener, \emph{Fourier transforms in the complex
  domain}.\hskip 1em plus 0.5em minus 0.4em\relax American Mathematical Soc.,
  1934, vol.~19.

\bibitem{Oppenheim:2015aa}
A.~V. Oppenheim and G.~C. Verghese, \emph{Signals, systems and
  inference}.\hskip 1em plus 0.5em minus 0.4em\relax Pearson, 2015.

\bibitem{Unser:1994aa}
M.~Unser and A.~Aldroubi, ``A general sampling theory for nonideal acquisition
  devices,'' \emph{IEEE Transactions on Signal Processing}, vol.~42, no.~11,
  pp. 2915--2925, 1994.

\bibitem{Unser:1997aa}
M.~Unser and J.~Zerubia, ``Generalized sampling: Stability and performance
  analysis,'' \emph{IEEE Transactions on Signal Processing}, vol.~45, pp.
  2941--2950, 1997.

\bibitem{Stanley:2012aa}
R.~P. Stanley, \emph{Enumerative combinatorics.}, ser. Cambridge studies in
  advanced mathematics: 49.\hskip 1em plus 0.5em minus 0.4em\relax New York :
  Cambridge University Press, 2012.

\bibitem{Rudin:1976aa}
W.~Rudin, \emph{Principles of mathematical analysis}, 3rd~ed.\hskip 1em plus
  0.5em minus 0.4em\relax McGraw-hill New York, 1976.

\bibitem{Papoulis:1977ab}
A.~Papoulis, \emph{Signal Analysis}.\hskip 1em plus 0.5em minus 0.4em\relax
  McGraw-Hill New York, 1977, vol. 191.

\bibitem{Polya:1915aa}
G.~P{\'o}lya, ``{\"U}ber ganzwertige ganze funktionen,'' \emph{Rendiconti del
  Circolo Matematico di Palermo (1884-1940)}, vol.~40, no.~1, pp. 1--16, 1915.

\bibitem{Polya:1920aa}
------, ``{\"U}ber ganze ganzwertige funktionen,'' \emph{Nachrichten von der
  Gesellschaft der Wissenschaften zu G{\"o}ttingen, Mathematisch-Physikalische
  Klasse}, vol. 1920, pp. 1--10, 1920.

\bibitem{Hardy:1917aa}
G.~H. Hardy, ``On a theorem of {M}r. {P}{\'o}lya,'' in \emph{Proc. Cam. Phil.
  Soc.}, vol.~19, 1917, pp. 60--63.

\bibitem{Whittaker:1935aa}
J.~M. Whittaker, \emph{Interpolatory function theory}.\hskip 1em plus 0.5em
  minus 0.4em\relax The University Press, 1935, vol.~33.

\bibitem{Selberg:1941aa}
A.~Selberg, ``\"{U}ber ganzwertige ganze transzendente {F}unktionen. {II},''
  \emph{Arch. Math. Naturvid.}, vol.~44, pp. 171--181, 1941.

\bibitem{Selberg:1941ab}
------, ``\"{U}ber ganzwertige ganze transzendente {F}unktionen,'' \emph{Arch.
  Math. Naturvid.}, vol.~44, pp. 45--52, 1941.

\bibitem{Pisot:1942aa}
C.~Pisot, ``\"{U}ber ganzwertige ganze {F}unktionen,'' \emph{Jber. Deutsch.
  Math. Verein.}, vol.~52, pp. 95--102, 1942.

\bibitem{Pisot:1946aa}
------, ``Sur les fonctions arithm\'etiques analytiques \`a croissance
  exponentielle,'' \emph{C. R. Acad. Sci. Paris}, vol. 222, pp. 988--990, 1946.

\bibitem{Pisot:1946ab}
------, ``Sur les fonctions analytiques arithm\'etiques et presque
  arithm\'etiques,'' \emph{C. R. Acad. Sci. Paris}, vol. 222, pp. 1027--1028,
  1946.

\bibitem{Buck:1948aa}
R.~C. Buck, ``Integral valued entire functions,'' \emph{Duke Math. J.},
  vol.~15, pp. 879--891, 1948.

\bibitem{Boas:1954aa}
R.~P. Boas, \emph{\BIBforeignlanguage{eng}{Entire Functions}}.\hskip 1em plus
  0.5em minus 0.4em\relax New York: Academic Press, 1954.

\end{thebibliography}

\end{document}